\documentclass[11pt]{article}

\usepackage[a4paper,margin=1.1in]{geometry}
\usepackage{authblk}
\usepackage{amsfonts}
\usepackage{amsthm}
\usepackage{amsmath}
\usepackage{graphicx}
\usepackage{comment}
\usepackage[usenames,dvipsnames]{xcolor}
\usepackage{enumerate}
\usepackage{extarrows}
\usepackage{hyperref}
\usepackage[round, authoryear]{natbib}
\usepackage{stackrel}
\usepackage{centernot}
\usepackage{caption}
\usepackage{subcaption}
\usepackage{fancyhdr}
\usepackage{bm}
\newtheorem{theorem}{Theorem}

\newtheorem{proposition}{Proposition}
\newtheorem{definition}{Definition}

\newtheorem{remark}{Remark}

\numberwithin{equation}{section}
\numberwithin{theorem}{section}
\numberwithin{corollary}{section}
\numberwithin{proposition}{section}
\numberwithin{lemma}{section}
\numberwithin{definition}{section}
\numberwithin{remark}{section}

\newtheorem{example0}{\sc Example}[subsection]

\makeatletter
\def\th@newremark{\th@remark\thm@headfont{\bfseries}}
\makeatletter
\def\nexto{\kern -0.54em}

\def\boxit#1{\vbox{\hrule\hbox{\vrule\kern6pt 
          \vbox{\kern6pt#1\kern6pt}\kern6pt\vrule}\hrule}}

\newcommand{\MMM}{{\cal M}}

\newcommand{\calP}{\mathcal{P}}

\newcommand{\BB}{\mathbb{B}}

\newcommand{\MM}{\mathbb{M}}

\newcommand{\R}{\Bbb{R}}

\newcommand{\N}{\Bbb{N}}

\newcommand{\wt}{\widetilde}

\newcommand{\dto}{\downarrow}

\newcommand{\be}{\begin{equation}}\newcommand{\ee}{\end{equation}}
\newcommand{\bea}{\begin{eqnarray}}\newcommand{\eea}{\end{eqnarray}}
\newcommand{\bean}{\begin{eqnarray*}}\newcommand{\eean}{\end{eqnarray*}}
\newcommand{\ben}{\begin{equation*}}\newcommand{\een}{\end{equation*}}
\newcommand{\ba}{\begin{aligned}}\newcommand{\ea}{\end{aligned}}

\newcommand{\BN}{\mathcal{BN}}

\newcommand{\bfJ}{{\bf J}}

\newcommand{\bfX}{{\bf X}}

\newcommand{\bfn}{{\bf n}}
\newcommand{\bfm}{{\bf m}}
\newcommand{\bfM}{{\bf M}}
\newcommand{\bfN}{{\bf N}}
\newcommand{\bfY}{{\bf Y}}
\newcommand{\bfe}{{\bf e}}

\newcommand{\PD}{\textbf{\rm PD}}
\newcommand{\PK}{\textbf{\rm PK}}
\newcommand{\PKr}{\textbf{\rm PK}^{(r)}}

\newcommand{\rmd}{{\rm d}}

\newcommand{\BBr}{\mathbb{B}^{(r)}}
\newcommand{\MMr}{\mathbb{M}^{(r)}}

\newcommand{\PP}{\mathbb{P}}
\newcommand{\EE}{\mathbb{E}}
\newcommand{\Tr}{T^{(r)}}

\begin{document}
%\title{Constructing Trees and  Coalescents for a Sampling Model Derived From a Negative Binomial Process}

\title{A Gibbs Sampling Scheme for  a Generalised Poisson-Kingman Class}

%\title{A Gibbs Sampling Scheme and Tree and Coalescent Constructions for Generalised Poisson-Dirichlet Processes}

\author[1]{Robert C. Griffiths}
 \affil{Department of Mathematics, Monash University, Clayton Victoria 3168, Australia.
 Email: bob.griffiths@monash.edu}

%\author[2]{Yuguang F. Ipsen}
%\affil{??  to get from Yuguang
%Email: }

\author[2]{Ross A. Maller}
\affil{Research School of Finance,  Actuarial Studies \&  Statistics,
Australian National University, Canberra, Australia.
Email: Ross.Maller@anu.edu.au
 }

\author[3]{Soudabeh Shemehsavar
\thanks{Email: Shemehsavar@khayam.ut.ac.ir}}
\affil[3]{School of Mathematics, Murdoch University, Perth, Western Australia;
School of Mathematics, Statistics \& Computer Sciences,
 University of Tehran. \newline
 Email: shemehsavar@ut.ac.ir (corresponding author)}

\maketitle
\vspace{-1cm}
\begin{abstract}\vspace{-0.03cm}
A  Bayesian
 nonparametric  method of James, Lijoi \& Prunster (2009)
 used to  predict future values of observations
 %, given those already observed, 
 from  normalized random
 measures with independent increments
 is modified 
 to a class of models 
  based on negative binomial processes
   for which the increments are not independent, but are independent conditional on an underlying gamma variable.
   Like in James et al., the new algorithm is formulated in terms of two variables, one a function of the past observations, and the other an updating by means of a new observation.
We outline an application of the procedure to population genetics,
 for the  construction of realisations of genealogical  trees and coalescents from samples of alleles.
%We want to construct trees and coalescents for $\PD_\alpha^{(r)}$....

\noindent {\small {\bf Keywords:}}

\end{abstract}

\section{Introduction}\label{Intro}\
 \cite{JamesLijoiPrunster2009}
 provided a comprehensive Bayesian
 nonparametric analysis of random probabilities,
 considered as  generalizations of the Dirichlet process,
  obtained by normalizing random
 measures with independent increments (NRMIs).
  They formulated  the posterior distribution of an NRMI as
 a mixture with  respect to the distribution of a specific latent variable, showed that  the posterior distribution 
 is again an NRMI, and derived from the ensuing   representation 
 a rule for predicting future values of observations, given those already observed. 
 This takes the form of  a mixture of two variables, one a function of the past observations, and the other an updating by means of the new observation.
 From this they  developed  a Gibbs sampling algorithm generalising
the Blackwell-MacQueen urn scheme
which provides a mechanism for simulating sample observations  
%$\bfX = \{X_1, \ldots, X_n\}$
 from the NRMI, 
 producing explicit and tractable expressions suitable for
 practical implementation.

 This procedure is closely related to 
 algorithms due to
 \cite{GTav1994b}
  for computing probability distributions of sample configurations of alleles in population genetics.
 \cite{GTav1994b} derive a recursion formula satisfied by the sampling probabilities and use it to construct a Markov chain 
%   with a set of absorbing states in such a way that the required sampling distribution is the mean of a functional of the process up to the absorption time. This 
 framework for simulating the likelihood of the data, given a set of parameter values. This approach corresponds to an importance sampling  procedure (\cite{FKYB1999}, \cite{SD2000}).
 %(Felsenstein et al., 1999; Stephens and Donnelly, 2000).
%  Stephens M., Donnelly P.
%Inference in molecular population genetics
%J. R. Stat. Soc. B, 62 (2000), pp. 605-635
%Felsenstein J., Kuhner M.K., Yamato J., Beerli P.
%Likelihoods on coalescents: A Monte Carlo sampling approach to inferring parameters from population samples of molecular data
%Seillier-Moiseiwitsch F. (Ed.), Statistics in Molecular Biology and Genetics, Institute of Mathematical Statistics and American Mathematics Society, Haywood, CA (1999), pp. 163-185
%De Iorio M., Griffiths R.C.
%Importance sampling on coalescent histories. I
%Adv. Appl. Probab., 36 (2004), pp. 417-433.
%De Iorio M., Griffiths R.C.
%Importance sampling on coalescent histories. II: Subdivided population models
%Adv. Appl. Probab., 36 (2004), pp. 434-454
%De Iorio M., Griffiths R.C., Leblois R., Rousset F.
%Stepwise mutation likelihood computation by sequential importance sampling in subdivided population models
%Theor. Popul. Biol., 68 (2005), pp. 41-53.
The recursive nature of the method results from looking
back in time in the coalescent tree to the most recent 
event in the past (either a coalescence or a mutation) to see what type of tree 
configuration 
there must have been 
at that time to produce the current  tree configuration.

   In the present paper we modify the
  \cite{JamesLijoiPrunster2009}  method 
  (also referring to \cite{LancelotJames2002, LancelotJames2005})
  to handle a wide class of models 
  based on negative binomial processes for which the increments are not independent, but are independent conditional on an underlying gamma variable.
  The formulae we derive  allow  adaptation of the  
  \cite{GTav1994b} methods
  %  These models also find application in gene sampling models 
 for the  construction of realisations of  trees and coalescents from observations, either simulated or actual, on the more general processes.

\subsection{The setup}\label{TS} \
Let  $(\xi_i)_{i\in\N}$,
$\N=\{1,2,\ldots\}$, 
 be a sequence of i.i.d. random variables on
$\R$ with a specified initial ``base'' distribution $G_0$.
Let  $(P_i)_{i\in\N}$ with $P_1 > P_2> \cdots  > 0$ be a random weight sequence on the infinite simplex $\nabla_\infty=
\{z_i\in [0, 1]: z_1 \ge z_2 \ge \cdots \ge 0,\, 
\sum_{i=1}^\infty z_i = 1\}$, 
%, thus $\sum_i P_i = 1$.
and $G$ a random discrete distribution generated from $(P_i)_{i\in\N}$ and  $(\xi_i)_{i\in\N}$.
With  $\delta_x$ denoting a point mass at $x\in\R$, 
$G$  can be written as
\be\label{ssm0}
G =\sum_{i\ge 1}P_i \delta_{\xi_i}.
\ee
%where the $(\xi_i)$  are independent and identically distributed (i.i.d.) random variables having a ``base" distribution $G_0$, independent of the weight sequence $(P_i)$. 
The data sample $\bfX = \{X_1, \ldots, X_n\}$
is obtained by drawing $n$ values from  $G$. 
Thus, conditionally given $G$, the $X_i$  are i.i.d. with distribution $G$. 
%The $X_i$ generate a random partition $\Pi$ consisting of non-overlapping ``blocks", by assigning $i$ and $j$ to the same block when $ X_i=X_j$. 
%We call $G$ a {\it species sampling model}; see for example \cite{ishwaran2003} for further discussion of these ideas.
%In the literature $G$ is known as a {\it species sampling model} (\cite{Pitman1995_SSB}).  
 In  the context of population genetics, the $X_i$ may represent the allele types of a particular locus on the chromosome. 
 
Due to the assumed discreteness of $G$, the observed data $\bfX$ in general contains ties, and can be reduced to a vector of unique values, $\bfY = \{Y_1, \ldots, Y_{K_n}\}$, and their corresponding multiplicities, 
$\bfN_{[n]}=(N_{[n]}^1, N_{[n]}^2, \ldots, N_{[n]}^{Kn})$. 
Here $K_n$ is a random variable increasing in $n$ representing the number of unique values -- clusters, or blocks -- in $\bfX$. 
Given $\bfX$, hence $K_n$, $K_n=k$, say, 
$ \bfN_{[n]}$ takes values among $k$-vectors of positive integers $\bfn = \{n_1, \ldots, n_k\}$
with $\sum_{i=1}^kn_i=n$.
 We consider the homogenous case where the $X_i$ are independent of the $P_i$.
 
 %The $X_i$ generate a random partition $\Pi$ consisting of non-overlapping ``blocks", by assigning $i$ and $j$ to the same block when $ X_i=X_j$. 

\subsection{Distributions on the infinite simplex:  
$\PK(\rho)$ and $\PK^{(r)}(\rho)$}
In the Poisson case, the $P_i$ in \eqref{ssm0} are  generated from 
%Basic to the construction  is 
a Poisson point process $\BB$ whose intensity measure  has density $\rho(s)$ satisfying 
\be\label{cond0}
\lim_{s \dto 0}\rho(s) = \infty, \quad 
\rho(s) <\infty \text{ for all } s> 0, \quad 
\text{ and } 
\int_{0}^1 s\rho(s) \rmd s<\infty. 
\ee
%Let $(P_i, i\ge 1)$ with $P_1 > P_2>P_3 > \ldots$ be a random sequence on the infinite simplex. We always have $\sum_i P_i = 1$. 
% a  L\'evy measure having density $\rho(x)$ on
%$(0,\infty)$ satisfying 
%\be\label{cond0}
%\lim_{x \dto 0}\rho(x) = \infty, \quad 
%\rho(x) <\infty \text{ for all } x> 0, \quad 
%\text{ and } 
%\int_{0}^1 x\rho(x) \rmd x<\infty.
%\ee
 We can write $\BB$ in the form $\BB = \sum_{i \ge 1} \delta_{S_i}$ with points
 $S_1> S_2>\cdots$.
 % be a Poisson point process having intensity measure $\rho$.
The  $P_i$ for the process $\BB$  are defined as the ratios
 of the ranked points of  $\BB$ normalised by their sum:
% to define a random sequence on the infinite simplex:
\be\label{PK}
 (P_i)_{i\in\N} =
 \Big(\frac{S_i}{\sum_{\ell\in \N}S_\ell}\Big) _{i\in\N}.
 %\sim \PK(\rho).
\ee
This kind of construction was first given by \cite{kingman1975}.
In his case the $S_i$ were the ranked jumps $\Delta S_1^{(1)} >\Delta S_1^{(2)}> \cdots$
of a subordinator $(S_t)_{t\ge 0}$ with L\'evy density $\rho(s)$ taken at time $1$.

The distributions of the $(P_i)_{i\in\N}$ thus constructed are said to be in  $\PK(\rho)$,  the {\it Poisson-Kingman class}.
These processes have been intensively studied for certain specifications of $\rho$. % such as we give below.
  We are interested in a generalised version which 
  we denote as  $\PKr(\rho)$.
To construct it, we modify  \eqref{PK} as follows.

\begin{definition}\label{d1}[Definition of $\BN(r, \rho)$]
Take a $\rho$ satisfying \eqref{cond0} and a parameter $r>0$.
Following \cite{Gregoire1984}, introduce  a negative binomial process, $\BBr$, defined in terms of  its Laplace functional:  
\be\label{lap_r}
\EE\big(e^{\BBr(f)}\big) =\Big(1+\int_{\R_+} (1-e^{-f(x)}) \rho(x)\rmd x\Big)^{-r},
\ee
 for any bounded nonnegative measurable function $f$ on $\R_+$. Such processes are said to be in $\BN(r, \rho)$.
 \end{definition}
 
  \noindent
 Basic properties of a negative binomial process are that numbers of points in non-overlapping Borel sets are negative binomially  (but  not independently) distributed.

 \eqref{lap_r} shows that we can regard a process in $\BN(r, \rho)$ to be  a Poisson point process with random intensity measure $\Gamma_r\rho$, where $\Gamma_r := \Gamma_{r,1}$, a Gamma random variable\footnote{We use the notation
  $\Gamma_{\beta, \gamma}$ for a Gamma random variable with parameter $(\beta, \gamma)$ having density 
$(\gamma x)e^{-\gamma x} (\gamma x)^{\beta -1}/\Gamma(\beta)$, $\beta, \gamma > 0$,  where $\Gamma(\cdot)$ is the gamma function.
} 
  with parameter $(r,1)$. 
Using  such a $\BBr$, we construct $\PKr(\rho)$ as follows.
  
\begin{definition}\label{d2}[Definition of $\PKr(\rho)$]
%Returning to \eqref{lap_r}, w
Write $\BBr = \sum_{i\in\N} J_i$ in terms of 
its  ranked points $J_1 > J_2 > J_3> \cdots$.
Analogous to \eqref{PK}, construct a random vector from the  ranked points in $\BBr$  normalised by their sum. The resulting
random sequence on the infinite simplex is said to have   distribution  $\PKr(\rho)$; thus,
\be\label{PKr}
 \Big(\frac{J_i}{\sum_{\ell\in \N} {J_\ell}}\Big) _{i\in\N} \sim \PKr(\rho), \quad r> 0.
\ee
\end{definition}

 L\'evy densities of particular interest are
\be\label{rho00}
\rho_\alpha (x) = \frac{\alpha}{\Gamma(1-\alpha)} x^{-\alpha-1} , \quad 
\rho_\theta (x) = \theta x^{-1} e^{-x}, \quad 
\rho_{G}(x) = \frac{\alpha}{\Gamma(1-\alpha)}x^{-\alpha-1} e^{-x}, x> 0,
\ee
where $\theta > 0$ and $ 0< \alpha <1$.
We can identify important members of the $\PK$ class in the $\PK^{(r)}$ framework. 
For example, the two-parameter Poisson-Dirichlet process 
$\PD(\alpha, \theta)$ (\cite{PY1997}) is obtained as
$\PD(\alpha, \theta)= \PK^{(\theta/\alpha)}(\rho_G)$.
Further, 
$\PD(\alpha, 0) = \PK^{(r)}(\rho_\alpha)$ for any $r>0$, 
and
$\PD(0, \theta)$ can be derived as a limiting case:
$\PD(0, \theta) = \lim_{\alpha \to 0}\PK^{(\theta/\alpha)}(\rho_G) = \lim_{r\to \infty} \PK^{(r)}(\rho_\theta/r)$
(\cite{PY1997},  \cite{IpsenMaller2017}).
%We can identify most members of the Poisson-Kingman class below, see \cite{IpsenMaller2017}, 
%\[
%\PD(\alpha, 0) = \PK^{(r)}(\rho_\alpha),\
% \text{for all } r > 0; 
%\quad \text{and} \quad 
%\PD(\alpha, \theta)= \PK^{(\theta/\alpha)}(\rho_G).
%\]
%Then, by \cite{PY1997}, $\PD(0, \theta)$ can be derived as a limiting case:
%$$\PD(0, \theta) = \lim_{\alpha \to 0}\PK^{(\theta/\alpha)}(\rho_G) = \lim_{r \to \infty} \PK^{(r)}\Big(\frac{1}{r}\rho_\theta\Big).$$

The choice 
 $\rho_R(x) = \alpha x^{-\alpha-1}{\bf 1}_{0<x\le 1}$
 in \eqref{cond0} gives the class  $\PD_\alpha^{(r)}$
 % we are  interested in for the present analysis.
studied in \cite{IpsenMallerShemehsavar2020a,
IpsenMallerShemehsavar2021} and  \cite{MallerShemehsavar2023b}
and related to the negative binomial process in \eqref{lap_r}.
% for various limiting forms of $\PD_\alpha^{(r)}$, and 
\cite{CZ2023} and \cite{LabadiZarepour2014} give further properties and  applications of  $\PD_\alpha^{(r)}$.
The distribution of  $\PD_\alpha^{(r)}$ converges to that of $\PD(\alpha, 0)$ as $r\dto 0$
(\cite{MallerShemehsavar2023b}).

\section{Posterior Distributions for  $\PKr(\rho)$}\label{pd}
\cite{LancelotJames2002, LancelotJames2005}  developed a framework to compute the likelihood and posterior distributions of partitioned data derived from an underlying Poisson point process. In particular, 
\cite{JamesLijoiPrunster2009} computed the posterior analysis for $G$ when $(P_i, i\ge 1)$ are in the 
Poisson-Kingman class. We use the same framework to derive the posterior distribution of $G$ when the weight sequence is distributed as $\PKr(\rho)$.
Like James et al. we make use of certain auxiliary random variables, denoted $(V_i)_{i\in\N}$  in our notation, to be specified later.

 By analogy with their approach, we define the following negative binomial point process augmented by the inclusion of an independent sequence  $\xi_i$:
\be\label{NB0}
\MMr(\cdot, \cdot) := \sum_{i} \delta_{(J_i, \xi_i)} 
\sim \BN(r, \rho, G_0),
\ee
where $(J_i)_{i\in\N}$ are the ranked  points of a $\BN(r, \rho)$ process $\BBr$,
and $(\xi_i)_{i\in\N}$ are i.i.d. with a
distribution $G_0$,
 independent of $(J_i)_{i\in\N}$.
% (In \cite{IpsenMallerShemehsavar2018} we chose the $U_i$ to be uniform on $[0,1]$.)
Note that $\MMr$ is no longer a Poisson random measure,
 but, conditional on a $\Gamma_r$ random variable, $\MMr$ is a Poisson point process with intensity measure $\Gamma_r \rho(\cdot)G_0(\cdot)$. 
We can relate $\MMr$ to the species sampling model $G$ in \eqref{ssm0} when $(P_i)_{i\in\N}$ is distributed as $\PK^{(r)}(\rho)$ by setting
\be\label{NBG}
G(\cdot) =  \frac{\mu(\cdot)}{\Tr}, \ \text{where} \
\mu(\cdot) = \int_{s>0} s \, \MMr(\rmd s, \cdot) \ \text{and} \
\Tr = \mu(\Omega) = \sum_{i\in\N} J_i.
\ee

%Fix the sample size $n$. Define a random variable 
%$\bfN_{[n]}=(N_{[n]}^1, N_{[n]}^2, \ldots, N_{[n]}^{Kn})$ which takes values as partitions of $[n]$ that is arranged in the order of appearance in the sampling process. Variations of $\bfN_{[n]}$ could be characterised in terms of  
%\begin{enumerate}[(i)]
%	\item variations in the number of distinct species, i.e. $\EE(K_n), \var(K_n)$
%	\item variations in the number of the $i$th species discovered given $i < K_n$, i.e. $\EE(N_{[n]}^i)$ and $\var(N_{[n]}^i)$
%	\item covariance of $N_{[n]}^i$ and $N_{[n]}^j$, given $i, j<K_n$.
%	\item variations in the $i$th largest species, given $i < K_n$.
%\end{enumerate}
%
%The density of vector $\bfN_{[n]}$ is known as the \emph{exchangeable partition probability function}, and a version of it is presented in \eqref{eppf}. However the form of the formulae is not so useful and the effect of $r$ is not explicit. Hence, in this paper, we aim to investigate the effect of $r$ numerically in the above framework and compare it with the $\PK$ class for a few notable examples.

We aim to use the methods of 
\cite{JamesLijoiPrunster2009} to derive the posterior distribution of the species sampling model $G$ %defined 
in \eqref{NBG} given observations $\bfX$. In the analysis, the $n$ observations in $\bfX$ appear in terms of their $k$ unique values $\bfY = \{Y_1, \ldots, Y_k\}$ and their respective multiplicities $ \bfN_{[n]}$.
% taking values among $k$-vectors of nonnegative integers $\bfn = \{n_1, \ldots, n_k\}$.

Our  first task is to compute the posterior distribution of the underlying augmented measure $\MMr$ in \eqref{NB0},
 given $\bfX$. 
In the following analysis, we fix $r > 0$ and base distribution $G_0$.  Define
\be\label{def1}
\psi(v) =1+ \int_{x>0} (1-e^{-vx})\rho(x)\rmd x
\ \text{and}  \
\pi_n(v) = (-1)^{n-1} \psi^{(n)}(v)
=\int_{x>0} x^n \rho(x) e^{-vx}\rmd x, \ v>0.
\ee
Conditional on $\bfX$, hence with $\bfY$, $k$ and $\bfn = \{n_1, \ldots, n_k\}$ given, define an auxiliary variable $V_n:=V_n(r)$ with conditional 
 density proportional to %$(V_n, \bfn)$ given by
\be\label{vn}
	g_r(v, \bfn):= \frac{r^{[k]}}{(\psi(v))^{r+k}} \cdot \,\frac{v^{n-1}}{\Gamma(n)}
	 \prod_{i=1}^{k}\pi_{n_i}(v),\ v>0,
\ee 
where $r^{[k]}=\Gamma(r+k)/\Gamma(r)$, $r>0$. 
Our main results stem from:
% version of Theorem 1 of  \cite{JamesLijoiPrunster2009}:

\begin{theorem}\label{postN}
Conditional on $V_n$ and $\bfX$, we have the following decomposition:
% the posterior distribution of $\MMr$ is 
\be\label{postN0}
\MM_\bfX^{(r, V_n)} = \wt \MM^{(r+k, V_n)} + \sum_{i=1}^{k} \delta_{J_{i}^{V_n}, Y_i},
\ee where,  for each $v>0$, 
\be\label{wtM}
\wt \MM^{(r+k, v)}\sim 
\BN\Big(r+k, \frac{e^{-vs}}{\psi(v)}\rho(s), G_0\Big),\ 
{\rm  independent\ of}\ \MMr,
\ee
 and the $J_{i}^{v}$, $1\le i\le k$,  are independent %i.i.d.??
  with  densities
	\be\label{dis:J}
	P\big(   J_{i}^{v} \in \rmd s\big)
	=
	\frac{s^{n_i}e^{-vs}\rho(s)\rmd s}	{\pi_{n_i}(v)}, \ s>0.
	\ee
\end{theorem}

\begin{remark}
{\rm
The process $\wt \MM$ in \eqref{wtM} is constructed from a negative binomial point process whereas the analogous process in \cite{JamesLijoiPrunster2009}
is from a Poisson point process.
A significant distinction is the appearance of the denominator  term $\psi(v)$ in \eqref{wtM} which of course changes the dynamics of the process appreciably.

% (Infinite Measure) 
The introduction of the auxiliary random variable $V_n$ conveniently decomposes the posterior structure of $\MMr$ into 
%a conventional form, that is 
an infinite measure part and a finite dimensional part. Conditional on 
$V_n$, the posterior distribution of a negative binomial process is 
again a negative binomial point process $\BN(r+k, \rho)$, noting in 
particular the change in parameter from $r$ to $r+k$, where $k$ is 
the current number of distinct values in $\bfX$.
}
\end{remark}

Using the posterior distribution for the underlying negative binomial process in Theorem  \ref{postN}, it is immediate to give the posterior distributions of $\mu$ and $G$ in \eqref{NBG}.
The following modifies Theorems 1 and 2
 of  \cite{JamesLijoiPrunster2009}:

\begin{theorem}\label{postU}
(i)\
Conditional on $\bfX$ and $V_n$, the posterior distribution of $\mu$ is 
\be\label{pmu}
\mu_\bfX^{V_n} = \mu^{V_n} + \sum_{i=1}^{k} J_i^{V_n} \delta_{Y_i}
\ee
where $\mu^{V_n}(\rmd x) = \int_{s>0} s \MM^{(r+k, V_n)}(\rmd s, \rmd x)$ and $(J_i^{V_n}, {Y_i})_{i=1, \ldots, k}$ are defined as in
Theorem \ref{postN}.
%\end{theorem}
%\begin{theorem}\label{postG}
%Conditional on $\bfX$ and $V_n$, 

 (ii)\ Conditional on $\bfX$ and $V_n$, the posterior distribution of $G$ is, 
 with $T^{V_n} = \mu^{V_n}(\Omega)$,
\be\label{pG}
G_\bfX^{V_n}(\rmd x) =  \frac{\mu^{V_n}(\rmd x)}{T^{V_n} + \sum_{i=1}^{k} J_i^{V_n}} +  \sum_{i=1}^{k} \frac{J_i^{V_n}}{T^{V_n} + \sum_{i=1}^{k} J_i^{V_n}} \delta_{Y_i} (\rmd x).
\ee
%where $\mu^{V_n}(\rmd x)$, $(J_i^{V_n}, {Y_i})_{i=1, \ldots, k}$ are defined as in Theorem \ref{postU}, with $T^{V_n} = \mu^{V_n}(\Omega)$.
\end{theorem}

%\begin{remark}{\rm \eqref{pG}is exactly analogous to
%the formula in  Theorem 2 of \cite{JamesLijoiPrunster2009}.
%}
%\end{remark}
% The random vector $\bfN_{[n]}=(N_{[n]}^1, N_{[n]}^2, \ldots, N_{[n]}^{K_n})$ takes values $\bfn$, i.e., of partitions of $[n]$ arranged in the order of appearance in the sampling process. 
% Variations of $\bfN_{[n]}$ could be characterised in terms of  
%\begin{enumerate}[(i)]
%	\item variations in the number of distinct species, i.e. $\EE(K_n), \var(K_n)$
%	\item variations in the number of the $i$th species discovered given $i < K_n$, i.e. $\EE(N_{[n]}^i)$ and $\var(N_{[n]}^i)$
%	\item covariance of $N_{[n]}^i$ and $N_{[n]}^j$, given $i, j<K_n$.
%	\item variations in the $i$th largest species, given $i < K_n$.
%\end{enumerate}

We obtain the marginal distribution of 
 $\bfN_{[n]}=(N_{[n]}^1, N_{[n]}^2, \ldots, N_{[n]}^{K_n})$   %$\bfn = (n_1, \ldots, n_k)$ 
 by integrating out $V_n$ in \eqref{vn}.
 This produces the \emph{exchangeable partition probability function (EPPF)}:
 	\be\label{eppf}
 	p(\bfn) = p(n_1, n_2, \ldots, n_k) = \int_v \Big[\prod_{i=1}^{k}\pi_{n_i}(v)\Big]
	 \times \frac{v^{n-1}}{\Gamma(n)}r^{[k]}\times (\psi(v))^{-(r+k)} \rmd v = \int g_r(v, \bfn)\rmd v.
 	\ee
%The distribution  of vector $\bfN_{[n]}$ is known as the \emph{exchangeable partition probability function} (EPPF), and a formula for it is in ??
%the Appendix (see \eqref{eppf}). 
%However the form of the formulae is not so useful and the effect of $r$ is not explicit. Hence, in this paper, we aim to investigate the effect of $r$ numerically in the above framework and compare it with the $\PK$ class for a few notable examples.

\section{Prediction Rules}\label{PR}
Here we state our main result, a formula for the predictive distribution of $X_{n+1}$ given the data $\bfX$. 
As we show in Section \ref{gibbs}, this can be used to construct an urn scheme for sampling analogous to the Chinese restaurant process commonly used to visualise these kinds of  procedure.
% related to $\PD(\theta)$ %(\cite{Hoppe1987}).
The following formula \eqref{pred0} is exactly analogous to formula (10) in Proposition  2 of \cite{JamesLijoiPrunster2009},
but again note
 the appearance of the denominator term $\psi(v)$ in \eqref{prednew} as a major distinction.
For  proof of Theorem \ref{predict} see the Appendix.

\begin{theorem}\label{predict}
Let $G$ be a species sampling model constructed from a negative binomial process with parameter $r$, intensity measure $\rho$ and base distribution $G_0$ as in \eqref{NBG}. Then the predictive distribution for $X_{n+1}$ given $\bfX$ has the same law as 
\be\label{pred0}
\PP(X_{n+1} \in \rmd x| \bfX) 
= \omega_0^{(n)} G_0(\rmd x) + \frac 1 n \sum_{i=1}^{k} \omega_i^{(n)} \delta_{Y_i} (\rmd x)
\ee
where 
\be\label{prednew}
\omega_0^{(n)} = 
\frac{r+k}{n} \int_0^\infty v\frac{\pi_1(v)}{\psi(v)} \,g_r(v, \bfn)\rmd v,
\ee and for $i = 1, \ldots, k$,
\be\label{predold}
\omega_i^{(n)} =  \int_0^\infty v  \frac{\pi_{n_i+1}(v)}{\pi_{n_i}(v)}\,g_r(v, \bfn)\rmd v.
\ee
\end{theorem}

\begin{remark}{\rm 
 	(i)\ The probability of forming a new cluster given $\bfX$ is $\omega_0^{(n)}$, which can be written as 
		\be\label{omega0b}
		\omega_0^{(n)} = \frac{r}{n} \int_0^\infty  \frac{ v\pi_1(v) (r+1)^{[k]}}{(\psi(v))^{r+1+k}} \frac{v^{n-1}}{\Gamma(n)} \prod_{i=1}^k \pi_{n_i}(v)\rmd v
		= \frac{r}{n} \int_0^\infty  v\pi_1(v) g_{r+1}(v, \bfn)\rmd v.
		\ee
		%	The effect of $r$ in \eqref{omega0b} can be factorised out as $h(r) = r^{[k]}(\psi(v))^{-r}$. 
%	We can investigate the properties of $h$ numerically and also $\omega_0^{(n}$ as a function of $r$. This would help with prior hyper-parameter selection. 
	%{\com what is the effect of $r$ in $\omega_0^{(n)}$? }
 
%		
%	\item(Finite Dimensional Part)
%	Conditional on $V_n$, the finite sum part appears as identical as in the Poisson case.
%The $J_i$ in 
%\eqref{dis:J} have the same distribution as in Theorem 1 of \cite{JamesLijoiPrunster2009}, hence the expression in \eqref{predold} is the same as that in Theorem 2 of \cite{JamesLijoiPrunster2009}, conditioning on their $U_n$.
%The effect of $r$ in our model is completely absorbed in $V_n$. 
%{\com what is the effect of $r$ in $\omega_i^{(n)}$? }

\noindent 
(ii)\
The {\it allele frequency spectrum} (AFS) is the $n$-vector
$\bfM_n=(M_{1n}, \ldots, M_{nn})$, 
where $M_{jn}= \sum_{i=1}^{k} 
{\bf 1}_{\{N_{[n]}^i=j\}}$, $1\le j\le n$.
We have $\sum_{j=1}^n j M_{jn}=n$
and  $\sum_{j=1}^nM_{jn}=k$.
 $\bfM_n$ takes values among $n$-vectors of nonnegative integers
$\bfm= (m_1,\ldots,m_n)$ satisfying  
\be\label{4.4a0}
A_{nk}:=\Big\{\bfm= (m_1,\ldots,m_n): m_j\ge 0, \, \sum_{j=1}^njm_j=n,\, \sum_{j=1}^nm_j=k\Big\}.   %,\ k\in \NN_n,\ n\in\NN,
\ee
%$\sum_{j=1}^njm_j=n$, while $K_n(\alpha,r)$ takes values $k=\sum_{j=1}^{n} m_j\in\NN_n:=\{1,2,\ldots,n\}$.

\noindent 
(iii)\ (Posterior Distribution of $K_n$)\
	Let $(D_i, i=1, \ldots, n)$ be independent Bernoulli($\omega_0^{(i)}$) random variables such that $D_i = 1$ when the $i$th individual sampled takes a different value from the previous $i$ labels. In the Chinese Restaurant process analogy, the $i$th customer sits at a different table with probability $\omega_0^{(i)}$. The posterior distribution of $K_n$ is $K_n^\bfX = 1+\sum_{i=1}^n D_i$.

From \eqref{vn} we can write $g_r(v, \bfn)$ alternatively in terms of the $m_j$ as:
\be\label{vn5}
	g_r(v, \bfm):= \frac{r^{[k]}}{(\psi(v))^{r+k}} \cdot \,\frac{v^{n-1}}{\Gamma(n)}
	 \prod_{j=1}^{n}(\pi_{j}(v))^{m_j},\ v>0.
\ee 	
}
\end{remark}

The next proposition as a check verifies that the probabilities in \eqref{pred0} and \eqref{prednew} add to the EPPF, and a   renormalisation of the EPPF gives the corresponding  partition distribution.

\begin{proposition}\label{3.1}\
Using the notation in Theorem \ref{predict},  we have
\be\label{nune}
 \omega_0^{(n)}+ \frac 1 n \sum_{i=1}^{k} \omega_i^{(n)} = \int_0^\infty g_r(v,\bfn)\rmd v;
 \ee
 and 
 \be\label{nunf}
 \sum_{k=1}^n \sum_{\bfm\in A_{nk}}
\frac{n!}{\prod_{j=1}^{n}j!^{m_j} m_j!}
 \int_0^\infty g_r(v,\bfm)\rmd v
 =1.
 \ee
\end{proposition}

%\section{A Sampling Algorithm}
%We wish to generate samples from the species sampling model $G$ derived from negative binomial processes as in \eqref{NB0}. In this section, for a fixed value $n$, we derive a Gibbs sampling scheme to produce a Markov chain $(\bfX^t, t\ge 0)$ of $n$-vectors, with stationary distribution equal to the likelihood of $\bfX = \{\bfn, \bfY\}$, which is proportional to 
%	\be\label{vnXjoint2}
%%	\prod_{i=1}^{k} G_0(Y_i) \int_v   \Big[\prod_{i=1}^{k}\pi_{n_i}(v)\Big]
%%	 \times \frac{v^{n-1}}{\Gamma(n)}r^{[k]}\times (\psi(v))^{-(r+k)}\rmd v = 
%	 \prod_{i=1}^{k} G_0(Y_i) \int_v  g_r(v, \bfn)\rmd v.
%	\ee
%This can be obtained by integrating out $V_n$ in the joint likelihood of $(\bfY, \bfn, V_n)$ in \eqref{vnXjoint}. 
%
%We obtain the conditional prediction rule by augmenting in \eqref{pred0} to get
%	\begin{align}\label{condPred}
%	\PP(X_{n+1} \in \rmd x,  V_n \in \rmd v| \bfX)
%	=\frac{r+k}{n} \frac{v\pi_1(v)}{\psi(v)} G_0(\rmd x)  + \frac 1 n v \sum_{i = 1}^k \frac{\pi_{n_i+1}(v)}{\pi_{n_i}(v)} \delta_{Y_i}(\rmd x).
%	\end{align}
%Dividing \eqref{condPred} by the joint density of $(V_n, \bfX)$ in \eqref{vn}, we get the conditional prediction rule,
%	\begin{align}\label{condPred2}
%	\PP(X_{n+1} \in \rmd x| \bfX, V_n=v)
%	\propto\frac{(r+k)\pi_1(v)}{\psi(v)} G_0(\rmd x)  + \sum_{j = 1}^k \frac{\pi_{n_j+1}(v)}{\pi_{n_i}(v)} \delta_{Y_i}(\rmd x).
%	\end{align}
%%Hence for $V_0(r)$, one can sample from a density proportional to $\frac{r\pi_1(v)}{\psi(v)}$. 
%The next subsection gives the specifics of the  sampling scheme outlined in Subsection \ref{TA}.

\section{A Gibbs sampling scheme for 
$\PD_\alpha^{(r)}$}\label{gibbs}
The following  Gibbs sampling scheme generalises the Blackwell-MacQueen urn scheme in \cite{JamesLijoiPrunster2009}, Section 3.3.
Choose a value of $\bfX = (X_1, \ldots, X_n)$ as follows.
\begin{enumerate}[\rm (i)]
	\item sample $V_0$ from density \eqref{vn} with $n=k=1$,
 which is  proportional to $\frac{r\pi_1(v)}{\psi(v)}$;
	\item sample $X_1$ from $\PP(\rmd X_1 | V_0)$, 
 which is  proportional to
	$\frac{r\pi_1(V_0)}{\psi(V_0)}G_0(\rmd x)$;
	\item\label{repeat} 
	for $\ell\ge 1$,
			\begin{enumerate}[\rm (a)]
			\item sample $V_\ell$ from $\PP(\rmd V_\ell|X_1, \ldots,X_{\ell})$, which is proportional to \eqref{vn} with $k$ replaced by $\ell$;
			\item sample $X_{\ell+1}$ from $\PP(\rmd X_{\ell+1} | X_{1}, \ldots, X_{\ell}, V_\ell)$ 
			%specified in \eqref{condPred2}.
using  the conditional prediction rule
$	\PP(X_{n+1} \in \rmd x| \bfX, V_n=v)$  proportional to
	\begin{align}\label{condPred2}
	\frac{(r+\ell)\pi_1(v)}{\psi(v)} G_0(\rmd x) 
	+ \sum_{i = 1}^\ell \frac{\pi_{n_i+1}(v)}{\pi_{n_i}(v)} \delta_{Y_i}(\rmd x);
	\end{align}
			\end{enumerate}
	\item repeat \eqref{repeat} until $X_n$ is reached.
	\end{enumerate}
	
	For analysis, collect the $X_\ell$ into  vector
	$\bfX = (X_1, \ldots, X_n)$, and  reduce $\bfX$ to a vector of unique values, $\bfY = \{Y_1, \ldots, Y_{K_n}\}$, and their corresponding multiplicities, $\bfn = \{n_1, \ldots, n_{K_n}\}$, with $\sum_{i=1}^{K_n} n_i=n$, as in Subsection \ref{TS}.
	 	Note that the EPPF  formula \eqref{eppf}
 	recovers Eq. (5.6) in
 	\cite{IpsenMallerShemehsavar2020a}
 	when $\rho(x)\rmd x=\alpha x^{-\alpha-1}\rmd x{\bf 1}_{0<x\le 1}$,
 	 so the algorithm in Subsection \ref{gibbs} does indeed produce samples from $\PD_\alpha^{(r)}$.
 	 
%%THE next relates to hierarchical models ?
%For each $t \ge 1$, generate a sample $\bfX^t = (X^t_1, \cdots, X^t_n)$ by, for $i=1, \ldots, n$, repeat
%\begin{enumerate}[\rm (a)]
%	\item sample $V^i_{n-1}$ from $\PP(\rmd V^i_{n-1}|\bfX^{t-1}_{(-i)})$, that is proportional to \eqref{vn}, where 
%	$$\bfX^{t-1}_{(-i)} = 
%(X_1^{t-1},\ldots,  X_{i-1}^{t-1}, X_{i+1}^{t-1}\ldots, X_n^{t-1});$$
%
%	\item sample $X^{t}_{i}$ from $\PP(\rmd X^{t}_i | \bfX^{t-1}_{(-i)}, V_{n-1}^i)$.
%\end{enumerate}

	\section{Constructing Trees and  Coalescents from a $\PD_\alpha^{(r)}$ Sample}\label{tree}
	A genealogical tree is a diagram that depicts the ancestral history and relationships between genes in a sample from within a species.
% from a population of interest. 
The tree's branches represent  lines of descent and the 
 branches'  tips are labelled with the  
 allele types. The pattern of branching reflects relationships between alleles; 
% and suggests how they may have evolved from  common ancestors. 
they are more closely related when there are fewer nodes on the path from one to another through the tree.
A  fundamental objective of population genetics analysis
is to construct genealogical trees from  DNA sequence data 
%sample genetic data 
%to display lines of evolutionary descent 
%from a common ancestor 
with the aim of developing probability models to
describe the ancestral processes involved.
	
%\cite{Hudson1991} describes many situations in which simulation of genealogical trees is useful. 
Good estimates of tree configurations 
    %ARGs    %ancestral recombination graphs 
are crucial for quantifying statistical uncertainty and estimating %population genetic 
parameters such as effective population size, mutation rate, allele age, etc. See \cite{Hudson1991}) for a general overview of the applications.
%The idea  is to construct a simulation of a coalescent tree, with times and branching order, and then
%superimpose the effects of mutation on this tree using an assumed Poisson nature of the mutation process. 
\cite{GTav1994b} describe a sequential importance sampling algorithm for calculating the probability of a sample of size $n$ under mutation and genetic drift in an unstructured population. 
	\cite{DEIG2004} developed an improved method of constructing sequential
 importance sampling proposal distributions on coalescent histories. 
% for computing the likelihood of a type configuration of genes in a sample of genes by simulation.
In this section we outline an extension of the \cite{GTav1994b} algorithm  to the $\PK^{(r)}$ class,
deriving a backwards recursion, depending on ancestral histories, for the sampling distribution.
% from $\PD_\alpha^{(r)}$. 
A time scale is introduced for the backwards processes.
 Such kinds of recursion are important in sequential importance sampling and duality with diffusion processes as illustrated for 
 the one and two parameter Poisson-Dirichlet processes
 in  \cite{GRSZ2023} and \cite{P2009}.
 
Consider a population consisting of a number of  genes, each of which has an
 infinite number of nucleotide sites, such that at each site there are two or more
 possible nucleotides, the original one and mutant ones. 
 In the infinitely-many-sites model of DNA sequences the ``type" of a gene is  given by a sequence of ls and 0s indicating at which sites a mutation has
 occurred in the gene or one of its ancestors. The reproductive mechanism  assumes in effect that each of the genes in a generation under consideration selects a parent
 gene at random (with replacement) from the previous generation. 
 With a designated  probability the offspring gene is
 of the same type as the parent gene; with the complementary probability a mutation occurs, in which case a 0 changes to a 1 at a site where no mutation
 has occurred in any gene. Which site is chosen for the mutation does not matter,  as sites are unordered.
 At an allele level sequences are recorded as being distinct if their DNA types are different, but the fine details of site mutations are not taken into account.
 See \cite{EG1987} for more details.

In a reverse procedure, for a given  set $D$ (descendents)  of input sequences in the generation under consideration, a genealogy for $D$ can be constructed backward in time using the evolutionary events of mutation and coalescence. 
Types are thought of as mutations and removing a singleton is like removing a mutation.
An  ancestral configuration can be regarded as the multiset of all sequences present at a particular point in time in a possible genealogy for $D$. 
It may be thought  of as a tree (reading forward in time) or a coalescent (reading backward in time).

The random vector $\bfN_{[n]}$ of Section \ref{pd},
 which takes values among partitions of $[n]$ is now interpreted as counting the  numbers of distinct alleles in a sample of size $n$ 
arranged in their order of appearance in the sample. 
Suppose we currently have an observation  $\bfn=(n_1,\ldots, n_k)$ on  $\bfN_{[n]}$, where $n_i$ is the number of allele type $i$ among the $n$ alleles, of which there are $k$ distinct types. 
%We have $n_i\ge 1$ and $\sum_{i=1}^kn_i =n$.
Let $\bfe_i =(0,0,\ldots,0, 1_i, 0, \ldots,  0,  0)$ 
be a unit  $n$-vector with a 1 in the $i$-th position.
Recalling  the notation in Theorem \ref{predict}, we argue that to obtain the configuration $\bfn $ there was a birth from $\bfn-\bfe_i$ with probability 
\be\label{b1}
\frac{n_i}{n}\frac{1}{n-1}\omega_i^{(n-1)}(\bfn  -\bfe_i ),
\ee
when $n_i>1$, 
or, for a singleton
($n_i=1$), with probability 
\be\label{b2}
\frac{1}{n} \omega_0^{(n-1)}(\bfn -\bfe_i ).
\ee
All the terms when $n_i=1$ are identical.  The factor of $n_i/n$ appears  because each individual in the current sample (of size $n$) has an equal probability of being the new individual from the immediate ancestral sample (of size $n-1$).

A full investigation of the coalescent structure and biological relevance of (\ref{b1}) and (\ref{b2}) for particular models is beyond the scope of this paper and yet to be determined.
But we can give some preliminary observations as follows.
From \eqref{b1} and \eqref{b2} we get  a fundamental equation
for the EPPF, $\text{p}(\bm{n})$, in  \eqref{eppf}:
\begin{align}\label{fundamental:00}
\text{p}(\bfn )
&=\sum_{n_i>1}\frac{n_i}{n}\cdot \frac{1}{n-1}\omega_i^{(n-1)}(\bfn  -\bfe_i )
+\sum_{\{i:n_i=1\}}\frac{1}{n}\omega_0^{(n-1)}(\bfn  -\bfe_i ).
\end{align}
 The second term counts the number of singletons in the sample,  multiplied by 
% $\sum_{i=1}^k {\bf 1}_{n_i=1} =m_1$,
 the probability of a singleton, $\omega_0^{(n)}$.
  All the terms in the second sum are identical.
  
Based on  (\ref{fundamental:00}) we define a continuous time Markov process 
%$(\bfn (t),t\geq 0)$, beginning at $\bfn $ and decreasing in $n$, by considering
\begin{align}\label{fundamental:01}
\frac{\rmd}{\rmd t}\text{p}(\bfn ,t)
&=-\varphi(\bfn )\text{p}(\bfn ,t)\nonumber\\
&+\varphi(\bfn )\sum_{n_i>1}\frac{n_i}{n(n-1)}\cdot \frac{\omega_i^{(n-1)}(\bfn  -\bfe_i )}{\text{p}(\bfn -\bfe_i )}\cdot \text{p}(\bfn -\bfe_i ,t)\nonumber\\
&+
\varphi(\bfn )\sum_{\{i:n_i=1\}}\frac{1}{n}\frac{\omega_0^{(n-1)}(\bfn  -\bfe_i )}{\text{p}(\bfn -\bfe_i )}\cdot
 \text{p}(\bfn -\bfe_i ,t),
%\label{fundamental:01}
%\\
%&:= -\varphi(\bfn )\text{p}(\bfn ,t)\nonumber \\
%&+\varphi(\bfn )\sum_{n_i>1}\gamma_i(\bfn ) \text{p}(\bfn -\bfe_i ,t)\\
%&+
%\varphi(\bfn )\sum_{\{i:n_i=1\}}\gamma_0(\bfn )
% \text{p}(\bfn -\bfe_i ,t).
% \label{fundamental:01a}
\end{align}
where $\varphi(\bfn )$ is a total rate for the configuration  $\bfn $.
%which needs to be chosen according to the model.
%(In the $\PD(\alpha,\theta)$ model  $\varphi(\bfn ) = \frac{1}{2}n(n+\theta-1)$ not depending on $\alpha$.)
%Clearly
%\[
%\sum_{n_i>1}\gamma_i(\bfn )\frac{\text{p}(\bfn -\bfe_i )}{\text{p}(\bfn )}
% + \sum_{\{i:n_i=1\}}\gamma_0(\bfn ) \frac{\text{p}(\bfn -\bfe_i )}{\text{p}(\bfn )}
% = 1.
%\]
Rescale $\text{p}(\bfn ,t)$ as
\begin{align*}
H(\bfn ,t) &= \frac{\text{p}(\bfn ,t)}{\text{p}(\bfn )}, \ t\ge 0,
\end{align*}
so that
\begin{equation}
\frac{\rmd}{\rmd t}H(\bfn ,t) = -\varphi(\bfn )H(\bfn ,t) 
 + \varphi(\bfn )\sum_{n_i>1}\frac{n_i}{n}H(\bfn -\bfe_i ,t) 
 + \varphi(\bfn )\sum_{\{i:n_i=1\}}\frac{1}{n}H(\bfn -\bfe_i ,t).
 \label{Heq:0}
\end{equation}
The rates in (\ref{Heq:0}) are $\varphi(\bfn )n_i/n$ when both $n_i>1$ or $n_i=1$, and of course we have
\[
\sum_{i=1}^k\varphi(\bfn )\frac{n_i}{n} = \varphi(\bfn ).
\]
$H(\bfn ,t)$ is a function of a process $(\bfn (t), t\geq 0)$, which begins at $\bfn (0)=\bfn $ and decreases back in time by coalescing within a group where $n_i > 1$ or removing a singleton (``mutation'') when $n_i  = 1$. The rates $\bfn  \to \bfn -\bfe_i $ are $\varphi(\bm{n})n_i/n$. The choice of $H(\bfn ,0)$ is important.
In the one and two parameter Poisson-Dirichlet processes the construction of $H(\bfn ,t)$ leads to a coalescent  process dual to the diffusions.

It is easy to convert p$(\bfn )$ into a partition distribution by taking
\[
{\cal P}(\bfn )
 = {n\choose \bfn }\frac{1}{m_1!m_2!\cdots m_n!}\text{p}(\bfn )
\]
where $m_\ell$ is the number of $n_j$'s equal to $\ell$, but we restrict ourselves mainly to the function $\text{p}$  for simplicity.

The ratios after calculation for $n_i>1$ and $n_i=1$ are
\begin{align}\label{fundamental:02a}
\frac{n_i}{n}\cdot\frac{1}{n-1}\frac{\omega_i^{(n-1)}(\bfn  -\bfe_i )}{\text{p}(\bfn -\bfe_i )}
&= \frac{n_i}{n(n-1)}\cdot
\frac{\int_0^\infty \frac{v^{n-1}}{(\psi(v))^{r+k}}\prod_{j=1}^k\pi_{n_j}(v)dv}
{\int_0^\infty \frac{v^{n-2}}{(\psi(v))^{r+k}}\prod_{j=1}^k\pi_{n_j-\delta_{ij}}(v)dv}
\end{align}
and
\begin{align}\label{fundamental:02b}
\frac{1}{n}\frac{\omega_0^{(n-1)}(\bfn  -\bfe_i )}{\text{p}(\bfn -\bfe_i )}
&=\frac{r+k-1}{n(n-1)}\cdot\frac{
\int_0^\infty\frac{v^{n-1}}{(\psi(v))^{r+k}}
\prod_{j=1}^k\pi_{n_j}(v)dv
}
{
\int_0^\infty\frac{v^{n-2}}{(\psi(v))^{r+k-1}}
\prod_{j=1,j\ne i}^k\pi_{n_j}(v)dv
}.
\end{align}
%where $\varphi(\bfn )$ is a total rate when $\bfn $, which can be chosen. 
Note that the numerator in (\ref{fundamental:02a}) 
is the same in both the ratios, so can be factored out of the 
sums in 
\eqref{fundamental:01}
%ould be taken out of the sums in (\ref{Heq:0}) 
and  connected with $\varphi(\bfn )$.
%The coefficients above agree after calculation with (\ref{fundamental:02}).\\
The two ratios in (\ref{fundamental:02a}) and
(\ref{fundamental:02b}) 
 can be written as
\begin{align*}
&\frac{n_i}{n}\cdot
\frac{
\int_0^\infty g(v;\bfn )dv
}
{
\int_0^\infty g(v;\bfn -\bfe_j)dv
}
\quad {\rm and} \quad
\frac{1}{n}\cdot
\frac{
\int_0^\infty g(v;\bfn )dv
}
{
\int_0^\infty g(v;\bfn -\bfe_j)dv
}.
\end{align*}
The two ratios are the same, however the meaning is different because a singleton is removed in the second ratio.
For comparison, in the two-parameter Pitman-Yor model the two terms in  (\ref{fundamental:02a}) and
(\ref{fundamental:02b})  are
\[
\frac{n_i(n_i-\alpha-1)}{n(\theta+n-1)}\text{~and~} \frac{\theta+(k-1)\alpha}{n(\theta+n-1)}
\]
 respectively. These have an important meaning in duality with the two-parameter diffusion.

\medskip\noindent{\large \bf Acknowledgement}\ 
We are very grateful to Yuguang Ipsen for her help in the initial stages of this investigation.

%\section{Numerical Investigation}
%\subsection{$\PK^{(r)}(\rho_\theta)$ v.s. $\PD(0, \theta)$}
%\subsection{$\PD_\alpha^{(r)}$ v.s. $\PD(\alpha, 0)$}

%In this case, $\rho(s) = \alpha s^{-\alpha-1}\bf 1_{s\le 1}$, some simplification can be obtained.
%
%\begin{align}
%	\psi(v) &= \alpha \int_0^1 (1-e^{-vx}) x^{-\alpha-1} \rmd x \nonumber \\
%			&= \alpha \int_0^1 \int_{y=0}^x ve^{-vy} \rmd y x^{-\alpha-1} \rmd x \nonumber \\
%			&= \alpha \int_{y = 0}^1 \int_{x=0}^y x^{-\alpha-1} \rmd x ve^{-vy} \rmd y  \nonumber \\
%			&=  \int_{y = 0}^1 (y^{-\alpha}-1) ve^{-vy} \rmd y  \nonumber \\
%			&=  \Gamma(1-\alpha) v^\alpha \Gamma(1-\alpha, v; 1) - (1-e^{-v})
%\end{align}
%where $\Gamma(1-\alpha, v; 1)$ is the incomplete gamma function defined as 
%\[
%\Gamma(\alpha, \beta; x) = \int_{0}^x \beta x^{-\beta y} \frac{(\beta y)^{\alpha-1}}{\Gamma(\alpha)} \rmd y.
%\]
%Hence 
%\[
%\psi(v) = \Gamma(1-\alpha) v^\alpha \Gamma(1-\alpha, v; 1) + e^{-v}.
%\]
%\subsection{Comparison to $\PD(\alpha, \theta)$?}

\section{Appendix: Proofs}
It is well known (e.g. \cite{LancelotJames2005})
that the conditional distribution of a Poisson point measure $\BB$ with intensity $\rho$ given that a point $S$ is observed equals $\BB_S = \wt \BB + \delta_S$,
 where $\wt \BB$ is  again a Poisson point process with the same distribution as, but independent of, $\BB$. 
We denote the law of $\BB$ with intensity $\rho$ by $\calP(\rmd \BB | \rho)$ and the Palm measure of $\BB$,  conditional on observing $S$, as $\calP_S (\rmd \BB | \rho)$, which is the law of $\BB_S$. 

The Palm measure for the negative binomial point process, $\BN(r, \rho)$, also has a neat form. 
Given that a point $S$ is observed in $\BBr$, the conditional distribution of $\BBr$, denoted by $\BBr_S$, can be written as 
$\BB^{(r)}_S = \wt \BB^{(r+1)} + \delta_S$, where $\wt \BB^{(r+1)}$ is distributed as $\BN(r+1, \rho)$,
 independent of $\BBr$ and $S$. See Prop. 4.3 in \cite{Gregoire1984} and
 Prop. 3.1 in \cite{IpsenMallerShemehsavar2020a}.
 %Prop 2.1 in \cite{trimfrenzy}.
 
\begin{proof}[Proof of Theorem \ref{postN}]
We make use of 
the characterisation that  a negative binomial process $\BN(r, \rho, G_0)$ is equal in distribution to a Poisson point process with intensity measure $\Gamma_r \rho \times G_0$, where $\Gamma_r$ is an independent Gamma(r,1) rv. 
Conditional on $\Gamma_r = w$, let $N = \sum_{i\ge 1} \delta_{J_i, Y_i}$ be a Poisson point process  with parameter $w\rho(s)G_0(x)$. 
We will use the  notation 
  $\calP(\rmd N | \Gamma_r\rho \cdot G_0)$  for the law of $N$
 and  $\calP(\rmd N |w\rho \cdot G_0)$ 
 for the law of $N$ conditional on $\Gamma=w$.
%In the homogenous case, we suppress the 

The observed data $(\bfn, \bfY)$ are samples from $\MMr$, which can be written with latent variables  
$(N, \bfJ,  \bfY, \Gamma_r)$.
%In analogy with 
We start with Eq. (23) of \cite{JamesLijoiPrunster2009}
(which we write with our variable $v$ rather than their $u$ and our initial distribution $G_0$ rather than their $H$)
to get the joint distribution of $(N,\bfJ, \bfY,V_n)$ 
as proportional to:
\be\label{jlp23}
\prod_{i=1}^{k}  s_i^{n_i}   e^{-vs_i}\rho(\rmd s_i|Y_i) G_0(\rmd Y_i)v^{n-1} e^{-\psi(v)} \calP_{\nu_v}(\rmd N).
\ee
In our case we take this conditional on $\Gamma_r=w$, 
let $\rho(\rmd s|Y)= w\rho(s)\rmd s$ 
(since we assume the homogeneous case)
and note that, with  $f(s, v) = vs$, 
\ben
 \EE(e^{-N(f)}|w\rho)
 = \exp\Big(\int_{s>0} (1-e^{-f(s,v)}) w\rho(s)\rmd s\Big)
  =e^{-\int_{s>0} (1-e^{-vs}) w\rho(s)\rmd s}
 =e^{-w\psi(v)}
\een
(so we substitute $e^{-w\psi(v)}$ for the $e^{-\psi(v)}$
in \eqref{jlp23}).
Eq. (2) of \cite{LancelotJames2005} gives a formula for a disintegration of a Poisson process $N$  with point $J$ in terms of the Palm distribution:
$N(\rmd J)  \calP(\rmd N |\nu)
=   \calP(\rmd N |\nu,J)\nu(\rmd J)$.
In our case $ \calP_{\nu_v}(\rmd N)$ in \eqref{jlp23} takes the form $ \calP_{\bfJ, \bfY}(\rmd N| we^{-vs}\rho(s) \cdot G_0)$.
Using this  we can write the joint likelihood of 
%\ben
%\{ \bfN_{[n]} = \{n_1, \ldots, n_k\},\,
%   \bfJ \in \{\rmd s_1, \ldots, \rmd s_k\},\,
%  V_n \in \rmd v,\,
%  \bfY \in \{\rmd Y_1, \ldots, \rmd Y_k\},\,
%    N,\, \Gamma_r \in \rmd w\}
%\een
$\{\bfN_{[n]},\, \bfJ, \bfY, \Gamma_r\}$
%$\{\bfN_{[n]},\, (J_i), V_n, \bfX, \Gamma_r\}$
as proportional to
\be\label{d}
\prod_{i=1}^{k}  
\big[ s_i^{n_i} e^{-vs_i} w \rho(s_i) \rmd s_i
  G_0(\rmd Y_i)\big]
\times
 \frac{v^{n-1}}{\Gamma(n)}% \EE(e^{-N(f)}|w\rho)
 e^{-w\psi(v)}
 \times 
 \calP_{\bfJ, \bfY}(\rmd N| we^{-vs}\rho(s) \cdot G_0)
 \times
 \frac{w^{r-1}}{\Gamma(r)}e^{-w}.
\ee
%where $f(s, x) = vs$. Here $\calP_{\bfJ, \bfY}(\rmd N| we^{-vs}\rho(s) \cdot G_0)$ is the Palm measure of a Poisson point process with intensity $w e^{-vs}\rho(s)G_0(\rmd x)$ given the points ${(J_i, Y_i), i=1, \ldots, k}$. 
%Recall
%\ben
% \EE(e^{-N(f)}|w\rho)
% = \exp\Big(\int_{s>0} (1-e^{-f(s,v)}) w\rho(s)\rmd s\Big)
% =e^{-w\psi(v)},
%\een
After some rearrangement,  integrate with respect to $\{ \bfJ, N, \Gamma_r\}$, 
to get the joint distribution of $\{\bfN_{[n]}, \bfY\}$ 
%$\{\bfN_{[n]}, \bfY, V_n\}$ 
as
 \bea\label{eqn1}
&&\Big[\prod_{i=1}^{k}\int_{s_i} s_i^{n_i} e^{-vs_i} \rho(s_i) 
\rmd s_i\Big]
 \times 
 \prod_{i=1}^{k} G_0(dY_i) \times 
\frac{v^{n-1}}{\Gamma(n)} \frac{\Gamma(r+k)}{\Gamma(r) } \cr
&&
\hskip2cm 
\times
\int_{\MMM} \int_{w}\calP_{\bfJ, \bfY}(\rmd N| we^{-vs}\rho(s) \cdot G_0)
\times
\frac{w^{r+k-1}}{\Gamma(r+k)}e^{-w\psi(v)}e^{-w}\rmd w.
\eea
The integral in the second line of \eqref{eqn1} equals
\[
(\psi(v))^{-(r+k)}\int_{\MMM} \calP_{\bfJ, \bfY}(\rmd N| \Gamma_{r+k, \psi(v)}  e^{-vs}\rho(s) \cdot G_0) 
\]
where $\Gamma_{r+k, \psi(v)}$ is an independent  Gamma$(r+k, \psi(v))$ random variable.

Note that a Poisson point process with intensity measure $\Gamma_{r, \beta} \rho$ is in $\BN(r, \rho/\beta)$, 
as we can  see by checking that  
the Laplace functionals are the same; namely, 
	\begin{align*}
  &E\exp\Big(-\int_0^\infty (1-e^{-f(x)})\Gamma_{r, \beta} \rho(x)\rmd x \Big) \\
= &\int_0^\infty \exp\Big(-\int_0^\infty (1-e^{-f(x)})y \rho(x)\rmd x \Big) \beta e^{-\beta y}\frac{(\beta y)^{r-1}}{\Gamma(r)}\rmd y \\
= &\beta^r  \int_0^\infty \exp\big(-y \int_0^\infty (1-e^{- f(x)}) \rho(x)\rmd x + \beta\big) \frac{y^{r-1}}{\Gamma(r)}\rmd y \\
= & \beta^r \Big(\beta +\int_0^\infty (1-e^{- f(x)}) \rho(x)\rmd x\Big)^{-r} 
= \Big(1+\frac{1}{\beta}\int_0^\infty (1-e^{- f(x)}) \rho(x)\rmd x\Big)^{-r},
\end{align*}
with  $f$ any bounded nonnegative measurable function on $\R_+$. 
This is of the form in \eqref{lap_r}.
Thus, we can write 
\be\label{alt2}
\calP(\rmd N | \Gamma_{r+k, \psi(v)}  e^{-vs}\rho(s) \cdot G_0) =  
\calP\Big(\rmd N | \Gamma_{r+k}  \frac{e^{-vs}\rho(s)}{\psi(v)} \cdot G_0\Big).
\ee

 Since the integral over the space of point measures $\MMM$ is $1$, we get from \eqref{eqn1} the joint distribution of $V_n$ and $\bfX$ as 
\be\label{vnXjoint}
\Big[\prod_{i=1}^{k}\pi_{n_i}(v)\Big]
 \times \prod_{i=1}^{k} G_0(\rmd Y_i) \times 
\frac{v^{n-1}}{\Gamma(n)}r^{[k]}\times (\psi(v))^{-(r+k)}
=
 \prod_{i=1}^{k} G_0(\rmd Y_i) \times g_r(v,\bfn),
\ee
where $ g_r(v,\bfn)$ is defined in \eqref{vn}.
Dividing this into \eqref{d} gives the joint density of $\{\bfJ, N, \Gamma_r\}$ conditional on $\{\bfX, V_n\}$ as 
\be\label{joint2}
\frac{(\psi(v))^{r+k}}{[r]^{k}}\prod_{i=1}^{k}\frac{s_i^{n_i}e^{-vs_i}\rho(s_i)}{\pi_{n_i}(v)}
\frac{w^{r+k-1}}{\Gamma(r)} e^{-w(\psi(v))}
\times
\calP_{\bfJ, \bfY}(\rmd N| w  e^{-vs}\rho(s)\cdot G_0).
\ee
Integrate out $w$ from \eqref{joint2} to  obtain the joint density of $\{ \bfJ, N\}$ given $\bfX$ and $V_n$ as 
\[
\prod_{i=1}^{k}\frac{s_i^{n_i}e^{-vs_i}\rho(s_i)}{\pi_{n_i}(v)} 
\times
 \calP_{\bfJ, \bfY}\Big(\rmd N| \Gamma_{r+k} \frac{e^{-vs}\rho(s)}{\psi(v)} \cdot G_0\Big). 
\]
Then note that the Palm measure of a Poisson measure with law $\calP_{\bfJ, \bfY}(\rmd N| w\frac{e^{-vs}\rho(s)}{\psi(v)} \cdot G_0 )$ has the same distribution as
\[
N^{v,w} + \sum_{i=1}^{k} \delta_{J_i^v, Y_i} , \quad \text{where  
$N^{v,w}$ has law $\calP\Big(\rmd N| w\frac{e^{-vs}\rho(s)}{\psi(v)} \cdot G_0\Big)$},
\]
and the $J_i^v$ have the density in \eqref{dis:J}.
From here we can read off the posterior distribution of $\{N, \bfJ\}$ given $\bfX$ and $V_n$. 
We obtain \eqref{postN0} by observing that $\calP\Big(\rmd N| \Gamma_{r+k} \frac{e^{-vs}\rho(s)}{\psi(v)} \cdot G_0\Big)$ is the law of $\BN(r+k, \frac{e^{-vs}\rho(s)}{\psi(v)}, G_0)$.
	\end{proof}
	
	%
%\subsection{The Auxiliary Variable $V_n$ and the EPPF}\label{append:Vn} 
%The joint likelihood of $(V_n,  \bfN_{[n]})$  in \eqref{vn} is 
%\be\label{vnjoint2}
%g_r(v, \bfn)=
% \Big[\prod_{i=1}^{k}\pi_{n_i}(v)\Big]
%	 \times \frac{v^{n-1}}{\Gamma(n)}r^{[k]}\times (\psi(v))^{-(r+k)}.
%\ee
% This can be obtained by integrating out $\bfY$ in \eqref{vnXjoint}.
%We can obtain the marginal distribution of 
% $\bfN_{[n]}$   %$\bfn = (n_1, \ldots, n_k)$ 
% by integrating out $V_n$ to get the EPPF
% %\emph{exchangeable partition probability function (EPPF)}, 
% 	\be\label{eppf}
% 	p(\bfn) = p(n_1, n_2, \ldots, n_k) = \int_v \Big[\prod_{i=1}^{k}\pi_{n_i}(v)\Big]
%	 \times \frac{v^{n-1}}{\Gamma(n)}r^{[k]}\times (\psi(v))^{-(r+k)} \rmd v = \int g_r(v, \bfn)\rmd v.
% 	\ee
% 	Note that this formula recovers Eq. (5.6) in
% 	\cite{IpsenMallerShemehsavar2020a}
% 	when $\rho(x)\rmd x=\alpha x^{-\alpha-1}\rmd x{\bf 1}_{0<x\le 1}$,
% 	 so the algorithm in Subsection \ref{gibbs} does indeed produce samples from $\PD_\alpha^{(r)}$.
% 	

\begin{proof}[Proof of Theorem \ref{predict}.]
Let $G$ be the species sampling model derived from the negative binomial process as in \eqref{NB0}. 
Let $\bfX$ be an independent sample of size $n$ from $G$, then,
following the proof of  Prop. 2, p.96, of \cite{JamesLijoiPrunster2009}, 
\begin{align*}
\PP(X_{n+1} \in \rmd x | \bfX) 
&=\int \PP(X_{n+1} \in \rmd x | G) \PP(G | \bfX) \nonumber \\
&= \int G(\rmd x) \PP(G | \bfX) 	\nonumber \\
&= \EE( G_{\bfX}(\rmd x) )
= \int_v \EE(G_{\bfX}^{v}(\rmd x)	) \, g_r(v, \bfn) \rmd v	. 	
\end{align*}
Here $G_\bfX$ is the posterior distribution of $G$ given $\bfX$, $G_\bfX^{V_n}$ is specified in Theorem \ref{postU} and $g_r(v, \bfn)$ is defined in \eqref{vn}.
Conditioning on $V_n=v$, by \eqref{pG}
%Theorem \ref{postU}, 
we have
\be\label{comp0}
\EE(G_{\bfX}^{v}(\rmd x)) 
= \EE\Bigg(\frac{\mu^{v}(\rmd x)}{T^v + \sum_{i=1}^{k} J_i^v}\Bigg)
+ \EE\Bigg(\frac{\sum_{i=1}^{k} J_i^v \delta_{Y_i}(\rmd x)}{T^v + \sum_{i=1}^{k} J_i^v}\Bigg):= I_1(v, \rmd x) + I_2(v, \rmd x),
\ee
with $\mu^v$ and  $T^v$ as in Theorem \ref{postU} and $(J_i^v,Y_i)$ as  in Proposition \ref{postN}.
We evaluate each component in \eqref{comp0} separately. 

Write as a shorthand $R :=\Gamma_{r+k, \psi(v)}$. Then by Theorem \ref{postU} and using \eqref{alt2} for the distribution of $\MM^{(r+k, v)}$, we get
\begin{align}\label{comp1a}
I_1(v, \rmd x)= &	\EE\Bigg(\frac{\mu^{v}(\rmd x)}{T^v + \sum_{i=1}^{k} J_i^v}\Bigg) \nonumber \\
=&	\EE  \int_\lambda e^{-\lambda T^v} \int_s s \MM^{(r+k, v)}( \rmd s, \rmd x) e^{-\lambda \sum_{i} J_i^v} \rmd \lambda \nonumber \\
=&   \int_w \int_\MMM \int_\lambda \int_s s  N(\rmd s, \rmd x) \cdot e^{-\lambda T^v}  \calP(\rmd N| w e^{-vs}\rho(s)\cdot G_0) \cdot \EE(e^{-\lambda \sum_{i} J_i^v} ) \rmd \lambda \PP(R \in \rmd w)
\end{align}
Here note that by Eq. (2) of
\cite{LancelotJames2005}, %Prop. 2.1, p.6,
\[
e^{-\lambda T^v}  \calP(\rmd N| w e^{-vs}\rho(s)\cdot G_0) = \EE(e^{-\lambda T^v}| w e^{-vs}\rho(s)\cdot G_0)  \calP(\rmd N| w e^{-(\lambda+v)s}\rho(s)\cdot G_0),
\]
where $\EE(e^{-\lambda T^v}| w e^{-vs}\rho(s)\cdot G_0) = \exp(-w\psi^{(v)}(\lambda))$, and 
\bean
\psi^{(v)}(\lambda) &=&
 \int (1-e^{-\lambda x})e^{-vx} \rho(x)\rmd x
  = \int (1-e^{-(\lambda+v) x}) \rho(x)\rmd x
  - \int (1-e^{-v x}) \rho(x)\rmd x \cr
&=&
 \psi(\lambda + v) - \psi(v).
\eean
Also 
\[
\EE(e^{-\lambda \sum_{i} J_i^v}) = \prod_{i=1}^{k} \frac{1}{\pi_{n_i}(v)}\int_{s} e^{-\lambda s} e^{-vs} s^{n_i}\rho(s)\rmd s
 =  \prod_{i=1}^{k} \frac{\pi_{n_i}(\lambda + v)}{\pi_{n_i}(v)}.
\]
Thus, \eqref{comp1a} equals
\begin{align}\label{comp1b}
&\int_w  \int_\lambda \int_s s \int_\MMM  N(\rmd s, \rmd x)  \calP(\rmd N| w e^{-(\lambda+v)s}\rho(s)G_0) \cdot e^{-w(\psi(\lambda+v) - \psi(v))}\cr
&\hskip7cm 
\times
 \prod_{i=1}^{k} \frac{\pi_{n_i}(\lambda + v)}{\pi_{n_i}(v)} \rmd \lambda \PP(R \in \rmd w).
\end{align}
Also note that 
\bean
 \int_s s \int_\MMM  N(\rmd s, \rmd x)  \calP(\rmd N| 
 w e^{-(\lambda+v)s}\rho(s)\cdot G_0) 
 &=&
 G_0(\rmd x) \cdot  w \int_s  e^{-(\lambda+v)s} s \rho(\rmd s)  \cr
 &=&
 w \pi_1(v+\lambda)G_0(\rmd x). 
\eean

Recalling the density $g_r(v, \bfn)$ in \eqref{vn} and that
 $R =\Gamma_{r+k, \psi(v)}$, we see that
\begin{align}\label{comp1c}
&\int_v I_1(v, \rmd x) g_r(v, \bfn) \rmd v 
 =
\int_v \EE\Bigg(\frac{\mu^{v}(\rmd x)}{T^v + \sum_{i=1}^{k} J_i^v}\Bigg)  g_r(v, \bfn)\rmd v\nonumber\\
 &=
 G_0(\rmd x) \int_v\int_\lambda  \prod_{i=1}^{k} \frac{\pi_{n_i}(v+\lambda)}{\pi_{n_i}(v)}
\int_w w \pi_1(v+\lambda) e^{-w(\psi(\lambda + v) - \psi(v))} \frac{w^{r+k-1}}{\Gamma(r+k)}e^{-w\psi(v)} \rmd w \, \rmd \lambda 
 \nonumber \\
& \hspace{3cm} \times \prod_{i=1}^k \pi_{n_i}(v)\frac{v^{n-1}}{\Gamma(n)}r^{[k]}  (\psi(v))^{r+k} (\psi(v))^{-(r+k)} \rmd v.
\end{align}
Collecting the $w$ terms together, this equals
\begin{align}\label{comp1cd}
&
G_0(\rmd x) \int_v\int_\lambda \pi_1(v+\lambda) \prod_{i=1}^{k} \pi_{n_i}(v+\lambda)
\int_w   e^{-w(\psi(\lambda + v) )} \frac{w^{r+k}}{\Gamma(r+k+1)} \rmd w \, \rmd \lambda \cr
&\hskip7cm
\times v^{n-1} \frac{\Gamma(r+k+1)}{\Gamma(r)\Gamma(n)}  \rmd v \nonumber \\ 
& =
G_0(\rmd x) \int_v\int_\lambda \pi_1(v+\lambda) \prod_{i=1}^{k} \pi_{n_i}(v+\lambda)
(\psi(\lambda+v))^{-(r+k)-1}\, \rmd \lambda 
\, v^{n-1} \frac{\Gamma(r+k+1)}{\Gamma(r)\Gamma(n)}  \rmd v.  \cr
&
\end{align}
Changing variable to $z = \lambda+ v$ and recalling \eqref{vn}, we get the RHS of \eqref{comp1cd} equal to
\begin{align*}
&G_0(\rmd x) \int_{z = 0}^\infty \Big( \int_{v = 0}^z v^{n-1}  \rmd v \Big)  \pi_1(z) \prod_{i=1}^{k} \pi_{n_i}(z)
(\psi(z))^{-(r+k)-1}\, 
\,  \frac{\Gamma(r+k+1)}{\Gamma(r)\Gamma(n)} \rmd z \nonumber \\ 
&=
 \frac{G_0(\rmd x)}{n}\int_{z = 0}^\infty \frac{z\pi_1(z)(r+k)}{\psi(z)} \Big\{ \frac{z^{n-1}}{\Gamma(n)} \prod_{i=1}^{k} \pi_{n_i}(z) (\psi(z))^{-(r+k)}r^{[k]}\, \Big\} \rmd z.
\end{align*}
Comparing the last expression to \eqref{prednew} we see that 
\[ 
\int_v I_1(v, \rmd x)g_r(v, \bfn)\rmd v=  \omega_0^{(n)} G_0(\rmd x).
\]

To get \eqref{predold}, we work with $I_2(v, \rmd x)$ in \eqref{comp0} in a similar way:
\begin{align}\label{comp2}
	I_2(v, \rmd x) 
	= &\EE\Bigg(\frac{\sum_{i=1}^{k} J_i^v \delta_{Y_i}(\rmd x)}{T^v + \sum_{i=1}^{K} J_i^v}\Bigg) \nonumber \\
	= & \int_\lambda \EE(e^{-\lambda T^v | \Gamma_{r+k, \psi(v)}e^{-v}\rho}) \EE\Big(e^{-\lambda \sum_{i=1}^k J_i^v} \sum_{i=1}^k J_i^v \delta_{Y_i}(\rmd x)\Big) \rmd \lambda.
\end{align}
Integrating out the $\Gamma_{r+k, \psi(v)}$ in a similar way as for
\eqref{comp1c}, we get  the first expectation 
in \eqref{comp2}
as
\begin{align*}
  & \EE(e^{-\lambda T^v | \Gamma_{r+k, \psi(v)}e^{-v}\rho}) \nonumber \\
= & \int_w \exp\Big(-\int_0^\infty (1-e^{-\lambda y})w e^{-vy}\rho(y)\rmd y\Big)	e^{-w\psi(v)} \frac{w^{r+k-1} (\psi(v))^{r+k}}{\Gamma(r+k)} \rmd w \nonumber \\
= & \Big(\frac{\psi(v)}{\psi(\lambda + v)}\Big)^{r+k}.
\end{align*}
Recalling the density of $(J_i^v)$ in \eqref{dis:J}, the second expectation in \eqref{comp2} can be simplified as 
\begin{align*}
  &	\EE\Big(e^{-\lambda \sum_{i=1}^k J_i^v} \sum_{i=1}^k J_i^v \delta_{Y_i}(\rmd x)\Big) \nonumber \\
= & \EE\Big( \sum_{i=1}^k e^{-\lambda \sum_{j} J_j^v} J_i^v \delta Y_i(\rmd x) \Big) \nonumber \\
= & \sum_{i=1}^k \int_{s_1} \cdots \int_{s_k} \prod_{j=1}^k e^{-\lambda s_j} s_i \delta_{Y_i} (\rmd x) \prod_{j=1}^k\frac{e^{-vs_j}s_j^{n_j} \rho(s_j)}{\pi_{n_j} (v)} \rmd s_1 \cdots \rmd s_k \nonumber \\
= & \sum_{i=1}^k \int_{s_1} \cdots \int_{s_k} \Big( \prod_{j=1}^k\frac{e^{-(v+ \lambda)s_j}s_j^{n_j} \rho(s_j)}{\pi_{n_j} (v)} \Big) s_i \delta_{Y_i} (\rmd x) \, \rmd s_1 \cdots \rmd s_k \nonumber \\
= & \sum_{i=1}^k \prod_{j\neq i}\frac{\pi_{n_j}(v+\lambda)}{\pi_{n_j} (v)}   \int_{s_i}    s_i \delta_{Y_i}  (\rmd x) \,  \frac{e^{-(v+\lambda)s_i} s_i^{n_i} \rho(s_i)}{\pi_{n_i}}\rmd s_i \nonumber \\
= & \sum_{i=1}^k \prod_{j\neq i}\frac{\pi_{n_j}(v+\lambda)}{\pi_{n_j} (v)}   \frac{\pi_{n_i+1}(\lambda +v)}{\pi_{n_i (v)}}\, \delta_{Y_i}  (\rmd x) \, \nonumber \\
= & \prod_{j=1}^k \frac{1}{\pi_{n_j} (v)}    \sum_{i=1}^k \Big( \pi_{n_i+1}(\lambda +v) \prod_{j\neq i}\pi_{n_j}(v+\lambda)  \Big) \, \delta_{Y_i}  (\rmd x).
\end{align*}
Putting these together and recalling $g_r(v, \bfn)$ in \eqref{vn}, we get 
\begin{align}\label{comp2c}
&\int_v I_2(v, \rmd x) g_r(v, \bfn)\rmd v\nonumber \\
&=
 \int_v\int_\lambda  \Big(\frac{\psi(v)}{\psi(\lambda + v)}\Big)^{r+k}
\prod_{j=1}^k \frac{1}{\pi_{n_j} (v)}    \sum_{i=1}^k \Big( \pi_{n_i+1}(\lambda +v) \prod_{j\neq i}\pi_{n_j}(v+\lambda)  \Big) \, \delta_{Y_i}  (\rmd x)  \rmd \lambda \nonumber \\
& \hspace{2in} \times \prod_{i=1}^k \pi_{n_i}(v)\frac{v^{n-1}}{\Gamma(n)}r^{[k]}   (\psi(v))^{-(r+k)} \rmd v \nonumber \\
&=
 \int_v\int_\lambda  
  \sum_{i=1}^k \Big( \pi_{n_i+1}(\lambda +v) \prod_{j\neq i}\pi_{n_j}(v+\lambda)  \Big) \, \delta_{Y_i}  (\rmd x)  \rmd \lambda   \times \frac{v^{n-1}}{\Gamma(n)}r^{[k]}  
  (\psi(\lambda + v))^{-(r+k)} \rmd v \nonumber \\
&=
 \int_{z=0}^\infty \Big( \int_{v= 0}^z v^{n-1} \rmd v  \Big)
  \sum_{i=1}^k \Big( \pi_{n_i+1}(z) \prod_{j\neq i}\pi_{n_j}(z)  \Big) \, \delta_{Y_i}  (\rmd x)  \times \frac{1}{\Gamma(n)}r^{[k]}  (\psi(z))^{-(r+k)} \rmd z  \nonumber \\
&= 
\int_{z=0}^\infty \frac{z}{n} \Big( \prod_{j=1}^k \pi_{n_j}(z)\frac{z^{n-1}}{\Gamma(n)}r^{[k]}  (\psi(z))^{-(r+k)}\Big)
   \sum_{i=1}^k  \Big( \frac{\pi_{n_i+1}(z)}{\pi_{n_i}(z)}   \Big) \, \delta_{Y_i}  (\rmd x)  \rmd z \nonumber \\
&=
\frac{1}{n} \sum_{i=1}^k  \int_z z     \frac{\pi_{n_i+1}(z)}{\pi_{n_i}(z)}   g_r(z, \bfn)\rmd z\,
 \delta_{Y_i}(\rmd x).
\end{align}
We recognise \eqref{comp2c} as the second term on the RHS of \eqref{pred0}. This completes the proof.	
\end{proof}

	 	 \begin{proof}[Proof of Proposition \ref{3.1}]
Using \eqref{vn} the RHS of \eqref{nune} is
\bean
&&
\int_0^\infty v
\Big( \frac{(r+k)}{n}
\frac{\pi_1(v)}{\psi(v)} 
 + \frac{1}{n} \sum_{i = 1}^k \frac{\pi_{n_i+1}(v)}{\pi_{n_i}(v)} \Big)
 g_r(v,\bfn) \rmd v\cr
 &&=
 \frac{r^{[k]}}{ n\Gamma(n)}
 \int_0^\infty \Big(
\frac{ (r+k)\pi_1(v)}{\psi(v)} 
 + \sum_{i = 1}^k \frac{\pi_{n_i+1}(v)}{\pi_{n_i}(v)}\Big)
\frac{v^{n} \prod_{i=1}^{k}\pi_{n_i}(v)}
{(\psi(v))^{r+k}}.
	 \eean
Leaving aside the factor of $r^{[k]}/n\Gamma(n)$, the integral term equals
	 \bean
	 &&
 \int_0^\infty \Big(
\frac{ (r+k)\rmd \psi(v)/\rmd v}{\psi(v)} 
 + \sum_{i = 1}^k \frac{1}{  \pi_{n_i}(v) }
 \big( -\frac{\rmd \pi_{n_i}}{\rmd v} \Big)
\frac{v^{n} \prod_{i=1}^{k}\pi_{n_i}(v)}{(\psi(v))^{r+k}}
\rmd v	\cr
	 &&
	 =
 \int_0^\infty \Big(
 (r+k)\frac{\rmd \log \psi(v)}{\rmd v}
 - \sum_{i = 1}^k \frac{\rmd\log \pi_{n_i}} {\rmd v}
 \Big)
\frac{v^{n} \prod_{i=1}^{k}\pi_{n_i}(v)}{(\psi(v))^{r+k}}
\rmd v	\cr
	 &&
	 =
 \int_0^\infty 
\frac{\rmd} {\rmd v}  \log\Big(
 \frac{(\psi(v))^{r+k}} { \prod_{i=1}^{k}\pi_{n_i}(v)}
 \Big)
\frac{v^{n}\prod_{i=1}^{k}\pi_{n_i}(v)}{(\psi(v))^{r+k}} 
\rmd v.
\eean
Reintroducing the factor of $r^{[k]}/n\Gamma(n)$
and defining 
\ben 
f(v) :=   \log\Big(
 \frac{(\psi(v))^{r+k}} { \prod_{i=1}^{k}\pi_{n_i}(v)}
 \Big),
 \een
the last expression can be written as
\ben
 \frac{r^{[k]}}{ n\Gamma(n)}
 \int_0^\infty \frac{\rmd f(v)} {\rmd v} v^n
 e^{-f(v)} \rmd v
 =
  \frac{r^{[k]}}{ \Gamma(n)}
 \int_0^\infty  v^{n-1} e^{-f(v)} \rmd v,
\een
with an  integration by parts at the last step. Substituting for $f(v)$, the RHS is seen to equal $ \int_0^\infty g_r(v,\bfn) \rmd v$.

(ii)\ Using \eqref{vn} and \eqref{vn5} we can write
\ben
\frac{n!}{\prod_{j=1}^{n}j!^{m_j} m_j!}
 g_r(v,\bfm)
 =
  \frac{n!r^{[k]}v^{n-1}}{\Gamma(n)(\psi(v))^{r+k}}
   \prod_{j=1}^{n}
\frac{(\pi_{j}(v))^{m_j}}{j!^{m_j} m_j!}.
\een
Let $f(x)=x^{-r}$, so that 
$f^{(k)}(x)= r^{[k]} (-1)^k x^{-r-k}$, and let
$\phi(v)=\psi(v)$.
Then $\pi_{j}(v)= (-1)^{j-1}\phi^{(j)}(v)$ and 
\bean
&&
\sum_{k=1}^n \sum_{\bfm\in A_{nk}}
\frac{n! }{\prod_{j=1}^{n}j!^{m_j} m_j!}
  g_r(v,\bfm)\cr
 &&=
   \frac{(-1)^n v^{n-1}}{\Gamma(n)}
     \sum_{k=1}^n
    \sum_{\bfm\in A_{nk}}
    n!f^{(k)}(\phi(v))
    \prod_{j=1}^{n} \frac{1}{m_j!}
    \Big(\frac{ (-1)^{j}\phi^{(j)}(v)}{j!}    \Big)^{m_j}.
\eean
Recall from $A_{nk}$ in \eqref{4.4a0} that $\sum_{j=1}^njm_j=n$ and $ \sum_{j=1}^nm_j=k$.
Then, after cancelling the factors of $-1$,  we recognise in the sum and product Fa\`a di Bruno's formula for 
$f^{(n)}(\phi(v))$, so the integral over $v$ of the expression equals
\ben
  \frac{1}{\Gamma(n)}
 \int_0^\infty v^{n-1} f^{(n)}(\phi(v))\rmd v,
 \een
 and this is seen to equal 1  after $n-1$ integrations by parts.
Hence  \eqref{nunf}.
	 \end{proof}

%\subsection{Some Explicit Computations of $\bfN_{[n]}$}
%
%Denote the EPPF of $\bfN_{[n]} = \bfn$, where $\bfn = \{n_1, n_2, \ldots, n_k\}$ by 
%$p(n_1, \ldots, n_k)$.
%Conditioning on $K_n > i$, and for each $1\le i-1+ j \le n $, the probability of the $i$th species discovered having $j$ representatives is 
%\begin{align*}
%	\PP(N^{i}_{[n]} = j) 
%	&= \sum_{k = i}^{n-j-i+1} \sum_{\sum_{h \neq i} n_h = n-j } \PP(N_{[n]} = (n_1, 
%	\ldots, n_{i-1}, j, n_{i+1}, \ldots, n_k) | K_n = k) \PP(K_n = k) 
%\end{align*}
%Thus, 
%\begin{align*}
%\PP(N_{[n]}^1 = j) &= \PP(K_n=1)p(j) {\bf 1}_{j = n} + 
%				    \PP(K_n = 2) p(j, n-j) \\ 
%				   &+ \sum_{n_2 + n_3 = n-j}\PP(K_n = 3) p(j, n_2, n_3) + 
%				    \cdots + \cdots +  \PP(K_n = n-j+1) p(j,1,1,\ldots, 1)	.
%\end{align*}
%To obtain $\EE(N_{[n]}^1)$ and $\var(N_{[n]}^1)$ this way, one needs to compute the EPPFs for all possible partitions and know the distribution of $K_n$. 

%\newpage
%\medskip\noindent{\bf Acknowledgement}\ 

{}

%\renewcommand{\bibfont}{\small}
%\bibliography{Library_Levy_2020}
%\bibliographystyle{newapa}

\end{document}